\pdfoutput=1 
\documentclass[12pt]{article}
\usepackage{euscript,amsfonts,amsmath,amssymb,color}
\usepackage{graphicx}
\usepackage{caption}
\usepackage{soul}
\usepackage{subcaption}
\usepackage{fullpage}
\usepackage[round,authoryear]{natbib}
\usepackage{array}
\usepackage{setspace}
\usepackage{soul}
\usepackage{comment}
\usepackage{amsmath}
\usepackage{amssymb}
\usepackage{amsthm}

\definecolor{webgreen}{rgb}{0,0.4,0}
\definecolor{webbrown}{rgb}{0.6,0,0}
\definecolor{purple}{rgb}{0.5,0,0.25}
\definecolor{darkblue}{rgb}{0,0,0.7}
\definecolor{darkred}{rgb}{0.7,0,0}
\definecolor{darkgreen}{rgb}{0,0.7,0}
\usepackage[pdfborder=false]{hyperref}
\hypersetup{colorlinks,citecolor=darkred,filecolor=black,linkcolor=darkblue,urlcolor=webgreen,pdfstartview=FitH}
\newcommand{\ignore}[1]{}
\newtheorem{lemma}{{\bf Lemma}}
\newtheorem{proposition}{{\bf Proposition}}
\newtheorem{corollary}{{\bf Corollary}}
\newtheorem{theorem}{{\bf Theorem}}
\newtheorem{defn}{{\bf Definition}}

\def \tc{\overline{\succ}_R}

\def \st{^{\star}}

\def\uv2{\underline{v}_2}
\def\ov2{\overline{v}_2}

\newcommand{\rarely}[1]{}

\begin{document}

	\allowdisplaybreaks

	\begin{titlepage}
		\title{\textbf{Choice by Rejection}\thanks{We are extremely grateful to Sean Horan for his detailed comments. This paper has benefited from discussions with Arunava Sen, Debasis Mishra, Rohan Dutta, Saptarshi Mukherjee and Abhinash Borah. We would also like to thank seminar participants at ACM EC 2021, DSE Winter school 2020 and  ISI-ISER Young Economists Workshop 2020 .}}
		\author{Bhavook Bhardwaj \thanks{Indian Statistical Institute, New Delhi ({\tt bhavook17r@isid.ac.in}).} \and Kriti Manocha \thanks{Indian Statistical Institute, New Delhi ({\tt kritim17r@isid.ac.in}).}}
		\maketitle

		\begin{abstract}
			We propose a boundedly rational model of choice where agents eliminate dominated alternatives using a transitive rationale before making a choice using a complete rationale. This model is related to the seminal two-stage model of \cite{manzini2007sequentially}, the Rational Shortlist Method (RSM). We analyze the model through \textit{reversals} in choice and provide its behavioral characterization. The procedure satisfies a weaker version of the \textit{Weak Axiom of Revealed Preference} (WARP) allowing for at most two reversals in choice in terms of set inclusion for any pair of alternatives. We show that the underlying rationales can be identified from the observable reversals in the choice. We also characterize a variant of this model in which both the rationales are transitive

			

		\end{abstract}
		\thispagestyle{empty}
	\end{titlepage}

	\section{Introduction}
	
	In the last two decades, various two-stage choice procedures have been proposed to rationalize systematic violations of the standard notion of rationality. In this paper we consider a new two-stage procedure of decision making in which a decision maker (DM) first shortlists a set of alternatives by \textit{rejecting} the set of \textit{minimal} alternatives with respect to the (first) rationale\footnote{We define rationale as a binary asymmetric relation}. By minimal, we mean an alternative which is dominated by some other alternative and does not dominate any other alternative. In the second stage, she chooses the maximum\footnote{Maximum alternative is the one which is not dominated by any other alternative with respect to the relation involved} alternative from the shortlisted set with respect to the (second) rationale.
	
	Rejecting the ``worst" alternatives before choosing is a natural way of making choices. Stochastic models of \cite{tversky1972choice}, \cite{dutta2020gradual} and deterministic models of
	\cite{apesteguia2013choice} and \cite{masatlioglu2007theory} discuss procedures of eliminating alternatives before making the choice. Such choice behaviors are often observed in real life too.
	
	\textbf{Example 1.} The editor of an Economics journal receives paper submissions and has the option of desk-rejecting before sending them to reviewers. Due to a large number of submissions, her rejection is based on the abstracts and she wishes to shortlist all \textit{reasonable} papers for a detailed review by the referees. It is natural to assume that her ranking over the papers might be incomplete. In order not to reject a possibly good quality paper, she chooses to eliminate only the set of \textit{minimal} papers before forwarding them to reviewers.

	\textbf{Example 2.} Economics department of a university is hiring for a faculty position. The selection committee has four applicants $x,y,z \text{ and } w$ to choose from. The applications are shortlisted for the interview on the basis of published work. The publications may not be comparable across sub-fields. Those applicants with no publication or with publication in \textit{lower valued} journals are \textit{rejected}. In the second round, the best candidate is chosen using an overall ranking (considering teaching experience, interview, conference presentations etc). While considering $x$ and $y$, $x$ is selected as a better candidate. When $z$ is also considered, the committee goes for $y$. If all the four candidates are compared, $x$ is selected.

	\begin{figure}[h]
		\centering
		\includegraphics[scale=0.28]{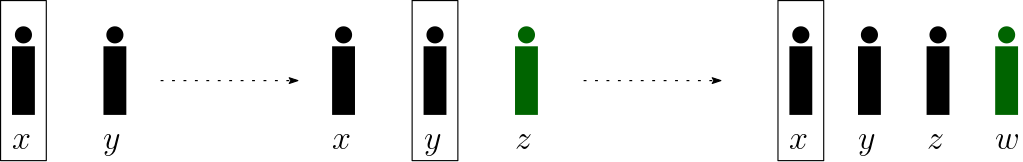}
	\end{figure}
	
	There is a large literature on two-stage choice procedures, well known as \textit{shortlisting procedures} (\cite{tyson2013behavioral}). These procedures rationalize boundedly rational behavior in different environments. The Rational Shortlist Method (RSM) (\cite{manzini2007sequentially}), Categorize then Choose (CTC) (\cite{manzini2012categorize}) and many related shortlisting procedures satisfy a weaker form of the \textit{Weak axiom of revealed preference} (WARP)\footnote{ \cite{samuelson1938note} showed that WARP alone behaviorally characterizes a choice function that is generated by maximizing by an underlying preference relation. It requires that for a pair of alternatives $x,y$, if $x$ is chosen in the presence of $y$, then $y$ cannot be chosen in the presence of $x$}. It requires that if an alternative $x$ is chosen in binary comparison with $y$, as well
	as in a set $S$ containing both $x$ and $y$, then $y$ should not be chosen in any ``intermediate'' set $S'$ between $\{x,y\}$ and $S$. Effectively, it allows for at most one reversal in terms of set inclusion for any pair of alternatives. Clearly, these models cannot explain scenarios like the ones described above where there is a reversal in choice from $x$ to $y$ and then $x$ again.
	
	We can rationalize this behavior using our model. For instance, publications of $y$ are in a different sub-field than those of $x,z \text{ and } w$  which are comparable with $z$ being the best and $w$ being the worst. Therefore, when only $x$ and $y$ are considered, both are shortlisted. If $z$ is also considered, $x$ is rejected on the basis of lower valued publications. If all the four are considered, $w$ being lowest ranked on the basis of publications, gets rejected. The overall ranking of the candidates, $x \succ y \succ z \succ w$, then rationalizes the final choices. Our model thus allows for a \textit{double reversal} in terms of set inclusion for a pair of alternatives. 
	
	In this paper, we formalize and analyze the model described above, called \textbf{Choice by Rejection (CBR)} . First, we axiomatically characterize the model where the first rationale is transitive and the second rationale is complete. Our analysis makes use of two types of choice reversals which we term as \textit{weak} and \textit{strong} reversals. The main axiom in our characterization restricts certain choice reversals. Second,  we show that first and second rationales which represent the data can be identified from the reversals. We use a \textit{small menu} property displayed by the choice function in identifying the class of \textit{CBR-representable} rationales. This property allows us to focus only on menus of pairs and triples. Third, we characterize a variant of CBR where the second rationale is restricted to be a linear order. Finally, we provide some results that relate CBR with existing shortlisting procedures like the RSM.

	\subsection{Related Literature} \label{lit}

	 The notion of ``eliminating" and choosing has been discussed in the literature. \cite{tversky1972choice} proposed a stochastic model where choice is analyzed as a probabilistic process of successive eliminations. In deterministic setting, \cite{apesteguia2013choice} proposed a procedure that involves sequential pairwise elimination of ``disliked" alternatives until only one alternative remains. \cite{masatlioglu2007theory} introduced a model of elimination wherein the DM eliminates those alternatives which are dominated by some ``comparable" alternative. Those  alternatives that cannot be eliminated by any of its comparables end up being chosen. Observe that the DM ends up choosing the \textit{maximal} set of comparable alternatives. On the other hand, in the two-stage models like RSM (\cite{manzini2007sequentially}), the \textit{maximal} set is shortlisted in the first stage. This can be understood as rejecting those alternatives which are dominated by some alternative. However, this entails large dependence on the first rationale since the second rationale is only used to choose one alternative among the small set of shortlisted maximal alternatives. This paper proposes a weaker form of domination in shortlisting where the second rationale has more deciding power. Note that choosing the maximal alternatives is equivalent to successively rejecting the set of \textit{minimal} alternatives i.e. those alternatives which are dominated by some alternative and do not dominate any other alternative. Successive elimination by the DM however can increase the cognitive load of shortlisting. It has been shown that often individuals deploy heuristics while making complex choices (\cite{gigerenzer1999fast}). We look at a simple heuristic instead in which DM rejects ``disliked" alternatives just once\footnote{If the ``rejection" is successive, in the limit, CBR is equivalent to the T$_1$SM model of \cite{matsuki2018choice}, a variant of RSM where the first rationale is transitive.}. A recent paper that is related to our model is by \cite{RePEc:ash:wpaper:29}. In their choice procedure, the DM shortlists alternatives by rejecting the worst alternative using a preference order. A detailed comparison with their work is done in section \ref{sec7}.

	The literature on boundedly rational choice procedures involves weakening of the standard notion of rationality i.e. WARP.
	One of the most well-known weakening of WARP is Weak-WARP (WWARP), first introduced in \cite{manzini2007sequentially} \footnote{Some of the models which directly use WWARP to characterize their models are \cite{manzini2007sequentially}, \cite{manzini2012categorize}, \cite{LOMBARDI200958}, \cite{cherepanov2013rationalization}, \cite{ehlers2008weakened} }.  We introduce a novel weakening of WWARP called R-WARP* which relates the two conditions using choice reversals. In terms of the number of reversals, it is well known that WWARP allows for at most one reversal between a pair of alternatives. R-WARP* extends this to at most two reversals which we call a \textit{double reversal}. Such behavior has been observed in different experimental settings (see \cite{manzini2010revealed}, \cite{teppan2009minimization}). The literature has attributed a single reversal to two well known effects called the \textit{compromise effect} and the \textit{attraction effect} which we will discuss later in the paper.  Similarly, \textit{two-compromise effect} and \textit{two-decoy effect} are observed as double reversal in choices (see \cite{tserenjigmid2019choosing}). Our paper gives a choice theoretic understanding of these effects. To analyze our model, we follow a technique similar to the one discussed in \cite{horan2016simple}. This involves viewing violations of rationality as choice reversals between pairs of alternatives. Our analysis relies primarily on two types of reversals permitted in this model which we term \textit{weak} and \textit{strong reversals} (discussed later). 
	
	The layout of the paper is as follows: Section \ref{sec2} discusses the model. Section \ref{sec4} provides axiomatic foundations of our model. Section \ref{sec5} discusses a variant of the model, \textit{Transitive}-CBR.  Section \ref{axiomsdiscussion} discusses choice reversals and their behavioral interpretations. Section \ref{Iden} provides results on identification of the model. Section \ref{sec7} provides some results relating our model to the literature and Section \ref{sec8} concludes.

		\section{Preliminaries} \label{sec2}

	Let $X$ be a finite set of alternatives and $\mathcal{P}(X)$ be the set of all non-empty subsets of $X$. The function $C:\mathcal{P}(X) \rightarrow X$ is a choice function that gives for any menu $S$ \footnote{$S \subset X \setminus$ $\phi$}, a unique alternative from $S$, i.e. $C(S) \in S$ and $|C(S)|=1$. Let $(R,P)$ denote a pair of rationales\footnote{$R \subset X \times X$ and $P \subset X \times X$} where $R$ is transitive and $P$ is complete. We define the set of \textit{minimal} alternatives with respect to $R$ from a menu $S$  as 
	\vspace{-0.5cm}
	$$ \min(S,R)=\{x \in S \ :\ \exists \ z \in S \ \text{s.t.} \ zRx \ \text{and} \ \nexists \ z' \in S \ \text{s.t.} \ xRz'\}$$
	Thus, an alternative is not minimal in a menu $S$ if and only if either (i) it is ``isolated" (not related to any other alternative with respect to $R$) or; (ii) there is at least one alternative which is ``dominated" by it with respect to $R$. 
	The idea of shortlisting in this paper relies on a one shot elimination of minimal alternatives before making the final choice as against shortlisting by selection of \textit{maximal} \footnote{Formally, the set of \textit{maximal} alternatives of choice problem is defined as $ \max(S,R)=\{y \in S | \ \nexists x \in S \ s.t \ xRy \}$} alternatives. 
	In our choice procedure, the DM first eliminates minimal alternatives using a selection criterion $R$ (first stage shortlisting) and then makes unique a choice from $S\setminus \min(S,R)$ by choosing the maximal alternative of the rationale $P$.  
\medskip
\begin{figure}[h]
\centering
\includegraphics[scale=0.8]{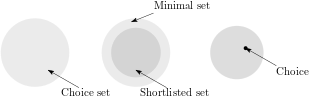}
\end{figure}
	
	\begin{defn}
		A choice function $C$ is \textbf{Choice by Rejection (CBR)} representable whenever there exists a pair of $(R,P)$, adsymmetric rationales with $R$ transitive (possibly incomplete) and $P$ complete such that 
		$$C(S)=\max(S \setminus \min(S,R),P)$$
\end{defn}
	
	Note that for any rationale $R$ on a set $S$, $ \max(S,R) \subseteq S \setminus \min(S,R)$. It indicates that for a given selection criterion, number of alternatives shortlisted in the first stage in CBR are at least as much as the number of alternatives shortlisted in RSM.\footnote{A choice function $C$ is RSM representable if it can be rationalized by an ordered pair of rationales $(P_1,P_2)$ such that $C(S) = \max (\max (S,P_1), P_2)$ }\\

\section{Behavioral Characterization} \label{sec4}

\subsection{Strong and Weak reversals}

 Observable choice reversals provide a succinct framework for analysis of boundedly rational models of choice. The characterization of the RSM model and the Transitive Shortlist Method (TSM)\footnote{TSM is a special case of the RSM model where both the rationales are transitive} by \cite{horan2016simple} is an important one in this regard. His characterization uses an interesting and easy to check consistency condition which is expressed using different types of choice reversals. In a similar manner, we categorize inconsistencies in choices in terms of choice reversals. We define two mutually exclusive reversals that help analyze our model and provide basis for our characterization. 

Consider three alternatives $x,y,z$ and a menu $S$ such that $ \{x,y\} \subseteq S$ and $z \notin S$. We say that the choice function $C$ displays an $(xy)$ reversal \textit{due} to $z$ if we observe the following choices
$$C(xy)=C(S)=x, \ \ C(S\cup \{z\})=y$$
Such $(xy)$ reversals can be categorized as \textit{weak} or \textit{strong} depending on whether reversal is due to an alternative which is either pairwise \textit{dominated} or \textit{dominates} $x$. We call the first one a \textbf{weak} $(xy)$ reversal where $x \succ_c z$. This reversal is a weak reversal (due to $z$) in the sense that the introduction of an apparently ``weak" alternative ($z$) shifts the choice from $x$ to $y$. The second type of reversal is called a \textbf{strong} $(xy)$ reversal where $z \succ_c x$. This reversal is a strong reversal (due to $z$) as the introduction of an apparently ``strong" alternative shifts the choice from $x$ to $y$. By definition, if $(xy)$ has a weak(strong) reversal due to $z$, then $(xy)$ cannot have a strong(weak) reversal due to $z$.\footnote {\cite{horan2016simple} describes Weak and Direct reversals in a similar spirit.   A choice function $C$ displays a Weak $(xy)$  reversal on $B \supset \{x,y\}$ if $C(xy)= x$ and $C(B) \neq C(B\setminus\{y\})$. $C$ displays a direct $(xy)$  reversal on $B \subseteq X \setminus\{x\}$ if $C(B)= y$ and $C(B \cup \{x\}) \notin \{x,y\}$.} We say that there is a reversal \textit{in the presence of} $x$ if it is already present in a menu on which a reversal happens. Formally, $C(S) = y,\  C(S \cup z) =w$ for some $S \ni x$. 

As it turns out, these reversals can provide us information about the first stage rationale. Intuitively, a reversal can occur when an alternative is in the minimal set for a given menu and upon addition of another alternative, it is ``pulled" out of the minimal set. Alternatively, an alternative can be ``pushed" into the minimal set upon addition of a new alternative. 

In order to capture all the information revealed by reversals, we define a relation $\succ_R$ on $X$ such that $x\succ_R y$ if and only if there is a:
\vspace{-0.2cm}
\begin{itemize}
	\setlength \itemsep{0.006cm}
	\item \textbf{weak} $(xy)$ reversal due to $w$ for some $w \in X$ or;
	\item \textbf{weak} $(wx)$ reversal due to $y$ for some $w \in X$ or;
	\item \textbf{strong} $(yw)$ reversal due to $x$ for some $w \in X$
\end{itemize}

It can be noted that by our definition of weak and strong reversals, $x\succ_R y$ would imply $x \succ_c y$, hence this relation is asymmetric. Also, $\succ_R$ may not be complete. The following example illustrates instances of weak and strong reversals and the resulting $\succ_R$.
\medskip

\textbf{Example 3:} Let $X = \{x,y,z, w\}$ and the choice function is as follows: 
{\small
	\begin{center}
		\begin{tabular}{|c|c|c|c|c|c|} 
			\hline 	
			{\textbf{S}} & {\textbf{C(S)}} & {\textbf{S}} & {\textbf{C(S)}} & {\textbf{S}} & {\textbf{C(S)}} \\
			\hline 
			$\{x,y\}$ & $x$ & $\{x,y,z\}$ & $y$ & $\{x,y,z,w\}$ & $x$ \\
			$\{x,z\}$ & $z$ & $\{x,y,w\}$ & $x$ &  & \\
			$\{x,w\}$ & $x$ &  $\{x,z,w\}$ & $x$ &  &\\
			$\{y,z\}$ & $y$ & $\{y,z,w\}$ & $y$ & & \\
			$\{y,w\}$ & $y$ &  &  & & \\
			$\{z,w\}$ & $z$ &  &  & & \\
			\hline 
		\end{tabular}
	\end{center} 
}
\medskip
It can be seen that in the choice function above, we have a
(i) \textbf{strong} $(xy)$ reversal due to $z$, and 
(ii) \textbf{weak} $(zx)$ reversal due to $w$. This generates the following $\succ_R$ 

\begin{figure}[h] 
	\centering
	\includegraphics[scale=0.35]{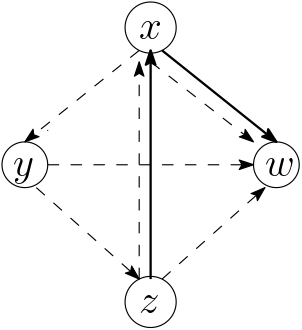}
	\caption{Dashed arrow indicates $\succ_c$ and solid arrow indicates $\succ_R$ } \label{psi}				
\end{figure}

An implication of CBR is that these reversals imply reversals on \textit{small} menus-- menus of size 2 and 3 -- as well, a result which we will prove later. This permits us to define $\succ_R$ solely based on choices from small menus. This is discussed in detail in section \ref{Iden}. Now, we are equipped to introduce the behavioral axioms.

\subsection{Axioms}
We now provide conditions on the choices of DM which guarantee that these choices are the result of DM choosing according to CBR. Our model is characterized by four behavioral properties (axioms) stated below. The first axiom is called \textit{Never Chosen}. It is a mild consistency condition first introduced in \cite{RePEc:ash:wpaper:29} and is related to the \textit{Always Chosen} property discussed in \cite{manzini2007sequentially}.\footnote{\textit{Always Chosen} is an intuitive property which says that if an alternative is chosen in pairwise comparisons with all alternatives of a menu, then it must be chosen from that menu} It says that for any menu, if an alternative is never chosen in a pairwise comparison with alternatives of that menu, then that alternative cannot be chosen in that menu. Formally, we define it as\\

\textbf{(A1)} \textbf{Never Chosen (NC)}: For all $S \in \mathcal{P}(X)$ and any $x \in S$,
\begin{center}
	$\forall \ y \in S \setminus \{x\}, \ C(xy)\neq x   \implies C(S) \neq x$
\end{center}
\medskip
Our second axiom is a novel weakening of the weak contraction consistency (WCC) axiom introduced in \cite{ehlers2008weakened}.\footnote{WCC states that if $C(S)=x$, then $C(S \setminus \{y\})=x$ for some $y \in S \setminus \{x\}$.} \\

\textbf{(A2)} \textbf{Weak Contraction Consistency$^*$ (WCC$^*$)}:
For any menu $S \supset \{x,y\}$ $$\text{If }C(S) \in \{x,y\}\text{, then there exists}\ z \ \in S \setminus \{x,y\} \text{ such that } C(S\setminus \{z\}) \in \{x,y\}$$
Intuitively, WCC ensures a ``path" of choices from  $\{x,y\}$ to the menu $S$ where either $x$ or $y$ is chosen. An interesting implication is that if there is a reversal from $x$ to $y$, there exists at least one intermediate set where addition of an alternative leads to the switch (see Figure \ref{WCC}). 
\begin{figure}[h]
	\centering
	\includegraphics[scale=0.45]{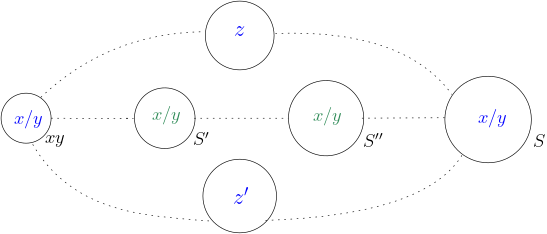}
	\caption{Existence of a path from $\{x,y\}$ to $S$ with choices belonging to $\{x,y\}$} \label{WCC}				
\end{figure}

In the spirit of the well known \textit{No Binary Cycles} (NBC) condition in the literature that restricts $\succ_c$ relation to be transitive, our next axiom prohibits cycles between only the pairs of alternatives related via $\succ_R$. Therefore, this condition can be viewed as a weaker form of NBC\\

\textbf{(A3)} \textbf{No binary cycles$^*$ (NBC$^*$)}:  For all $x_1, \ldots, x_n \in X$, $$x_1 \succ_R x_2, \ x_2 \succ_R x_3, \ldots , \ x_{n-1} \succ_R x_n \implies x_1 \succ_c x_n $$

Our last axiom is the classic congruence condition required in any shortlisting procedure as discussed in \cite{tyson2013behavioral}. Intuitively, it requires that if an alternative $x$ is chosen in the presence of another alternative $y$ where $y$ is not ``dominated", then $y$ cannot be chosen in presence of $x$ whenever $x$ is not ``dominated". This domination can be captured using $\tc$ relation (We denote transitive closure of $\succ_R$ as $\tc$). We call this condition Reject-WARP (R-WARP).\\

\textbf{(A4)} \textbf{Reject-WARP (R-WARP)}: For any alternatives $x$ and $y$  and menus $S$ $S'$  such that  $\{x,y\} \subseteq S , S'$ 
\begin{center}
	$\text{If } y \notin \min(S, \tc),\ \text{and} \  C(S)=x ,\ \text{then} \ x \notin \min(S', \tc) \implies C(S') \neq y$
\end{center}

Our model can be behaviorally characterized using the above discussed axioms. The main result of our paper is as follows
\begin{theorem} \label{thm1}
	A choice function $C$ is CBR representable if and only if it satisfies (A1)-(A4)
\end{theorem}

\paragraph{Outline of the proof:} 
WCC$^*$ and R-WARP imply that any choice reversal will be associated with an alternative $z$ such that the reversal is due to $z$. Hence, any reversal will be either a weak or a strong reversal. The axioms imply an exclusivity property which restricts the choice function such that if it displays a weak(strong) reversal for a pair of alternatives, then it cannot display a strong(weak) reversal.
NC and NBC$^*$ further impose restriction on $\succ_R$ when the choice function displays a strong reversal for a pair of alternatives. A small menu property helps us view all the reversals displayed by the choice functions on small menus i.e. menus of size 2 and 3. This enables us to construct rationales for representation of the choice data. 

	\section{Transitive-CBR} \label{sec5}

In this section, we discuss a variant of our model in which we restrict the second rationale to be a preference order. We call this variant \textit{Transitive-CBR}. This model is related to \cite{RePEc:ash:wpaper:29} as it relaxes completeness of the first rationale from their model. It can be seen as a natural generalization of their dual self model. It may be argued that the \textit{`should'}- self interpretation of the first rationale can display instances of indecisiveness which is precisely reflected by dropping their assumption of completeness. 

We can characterize this model by generalizing R-WARP to \textit{R-SARP} which is defined as: \\	
\textbf{(A4')} \textit{R-SARP}: 
For all $S_{1}, \ldots, S_{n} \in \mathcal{P}(X)$ and distinct $x_{1}, \ldots, x_{n} \in X:$
$$
\text{If} \ 	x_{i+1} \notin \min(S_i, \tc), \ C(S_{i})=x_{i} \text{ \ for \ }  i=1, \ldots, n-1, \text { then \ }$$ $$x_{1} \notin \min(S_n, \tc) \implies C\left(S_{n}\right) \neq x_{n}
$$
$\text{If } y \notin \min(S, \succ_R),\ \text{and} \  C(S)=x ,\ \text{then} \ x \notin \min(S', \succ_R) \Rightarrow C(S') \neq y$
It turns out that a characterization of \textit{Transitive-CBR} requires no more than this generalization of R-WARP to any arbitrary chain of alternatives. The characterization is then given by the following result

\begin{theorem}
	A choice function $C$ is a Transitive-CBR representable if and only it satisfies (A1)-A(3) and (A4') 
\end{theorem}
The proof can be found in the Appendix.

\section{Discussion on Choice Reversals}\label{axiomsdiscussion}

Rational choice theory does not allow for reversals i.e. the choice of an alternative $x$ when $y$ is available in a menu and the choice of $y$ when $x$ is available in a different menu. The literature is replete with empirical evidence displaying such reversals. Two prominent behavioral explanations of such reversals have been the \textit{compromise effect}  and the \textit{attraction effect} which is also  popularly known as the \textit{decoy effect}. The compromise effect  first discussed in \cite{simonson1989choice}  says that individuals avoid ``extreme" alternatives and ``compromise" for non-extreme alternatives. The idea is that addition of an alternative to a menu makes the previously chosen alternative appear ``extreme". Hence the choice shifts to an alternative which was not previously chosen, causing a reversal. The attraction effect first discussed in \cite{huber1982adding} on the other hand says that the addition of an alternative to a menu acts as a ``decoy" for an alternative that was previously not chosen. For alternatives $x$, $y$ and $z$, both the effects would be reflected behaviorally as $$C(\{x,y\})= x \ \ \ \text{and} \ \ \ C(\{x,y,z\}) = y$$
with $z$ acting as alternative that makes $x$ appear ``extreme" in the compromise effect and $z$ acting as a ``decoy" for $y$ in the decoy effect.   
We extend the idea above to what we call a \textit{single reversal}. Denote by $\succ_c$ the pairwise relation such that $x \succ_c y $ if and only if $C(\{x,y\})=x$ (we will abuse notation and use $(xy)$ and $\{x,y\}$ interchangeably). We now define a \textit{single reversal} as

\begin{defn}{$(xy)$ \textit{single reversal}:}
	If $x \succ_c y$ and there exists $ S \supset \{x,y\}$ such that $C(S)=y$ then for $S' \supset S \supset \{x,y\}$, $C(S') \neq x$  
\end{defn}
\begin{figure}[h] 
	\centering
	\includegraphics[scale=0.6]{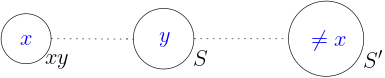}
	\caption{$(xy)$ single reversal} \label{WWARP}				
\end{figure}
The above definition permits for at most one reversal with respect to a pair $(xy)$ in terms of set inclusion. It is easy to see that if a choice function satisfies WARP, then for a pair of alternatives $(xy)$ , $x \succ_c y$ would imply that $y$ can never be chosen from any menu that contains $x$. Expressed in terms of reversals, WARP allows for no reversal in choices between $x$ and $y$ along any sequence of sets (containing $x$ and $y$) ordered by set inclusion. Whereas WWARP allows for only  \textit{single reversal} in choices. 

A natural implication of the compromise effect  and the decoy effect are what \cite{tserenjigmid2019choosing} calls the \textit{two-compromise effect}  and the \textit{two-decoy effect}. In the case of the two-compromise effect, the argument is that an addition of the fourth alternative $w$ to a menu would make $x$ no longer appear an ``extreme" alternative and the choice would revert to $x$. In case of the two-decoy effect, $w$ would act as a ``decoy" for $x$, nullifying the decoy effect of $z$ for $y$.  Again, both the effects would be reflected behaviorally as $$C(\{x,y\})= x \ \ \ \text{and} \ \ \ C(\{x,y,z\}) = y \ \ \ \text{and} \ \ \ C(\{x,y,z, w\}) = x$$

In a similar manner as a \textit{single reversal}, we extend the above idea to what we call a \textit{double reversal} defined as

\begin{defn}{	$(xy)$ \textit{double reversal}:}
	If $x \succ_c y$ and there exists $S' \supset  S \supset \{x,y\}$ such that $C(S)=y, \ C(S')=x$ then  for  $S'' \supset S' \supset \{x,y\} $, $C(S'') \neq y$
\end{defn}

\begin{figure}[h] 
	\centering
	\includegraphics[scale=0.6]{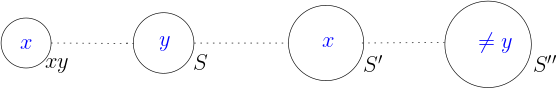}
	\caption{$(xy)$ double reversal} \label{R-WARP}				
\end{figure}

There is experimental evidence of double reversals (see \cite{tserenjigmid2019choosing}, \cite{manzini2010revealed} , \cite{teppan2009minimization}). We can see from example in the introduction that CBR allows for a double reversal and this is what differentiates CBR from other shortlisting models in the literature.\footnote{To the best of our knowledge, no shortlisting procedure disscussed in the literature allows for \textit{double reversals}}

Choice reversals provide a framework to relate our axioms to some well-known axioms in the literature. An interesting implication of R-WARP and WCC$^*$ is that for any pair $(xy)$, there can be no more than two reversals. So for a $(xy)$ reversal from $S$ to $S'$, we can identify a menu $T$ and alternative $z$, such that $S \subseteq T \subset S'$, $C(T)=x$ and $C(T \cup \{z\})=y$, and choice is $x$ for all sets in a ``path" between $S$ and $T$, and choice is $y$ in a ``path" between $T \cup \{z\}$ and $S'$. Similarly, for a \textit{double reversal}, we can identify two menus where addition of an alternative leads to a reversal in the ``path". Thus, an $(xy)$ double reversal in the choice is associated with two alternatives $z_1$ and $z_2$ due to which the reversal takes place. The above axioms imply a weaker version of {WWARP which we call R-WARP*. This condition restricts the number of reversals in any pair  to at most two.
	
	\begin{defn}{\textbf{R-WARP*}:}
		For all menus $S,S', S''$ such that $\{x,y\} \subset S' \subset S \subset S''$ 
		\begin{center}
			$C(S) = C\{x,y\} = x$ and $C(S') = y$ implies $C(S'') \neq y$
		\end{center}
	\end{defn}
	The above discussed restriction can be summarized by the following result
	\begin{lemma}\label{rwwarp}
		If $C$ satisfies R-WARP and WCC$^*$, then it satisfies R-WARP*
	\end{lemma}

	Another interesting implication of the axioms above is a condition which imposes clear limitations on the possibility of certain simultaneous weak and strong reversals. For a given weak reversal it precludes certain strong reversals and vice-versa. This is captured in a property which we call \textit{Exclusivity}.\footnote{This is closely related to the Exclusivity condition of \cite{horan2016simple}} It allows for only one type of reversal between a pair due to any alternative.

	\begin{defn} {\textit{Exclusivity}}: 
		For any pair of alternatives $(xy)$, either:
		\begin{itemize}
			\item $C$ displays no \textbf{weak} $(xy)$ reversal; or
			\item $C$ displays no \textbf{strong} $(xy)$ reversal
			
		\end{itemize}
	\end{defn}
	For any pair of alternatives, this condition precludes choice behavior which exhibits both types of reversals, strong and weak. Put differently, the possibility of strong reversals for a given pair of alternatives is ruled out by observing a single weak reversal for that pair (and vice versa). A corollary of the above result is the following result, which we use in the proof of Theorem \ref{thm1} 
	\begin{lemma}\label{corr}
		If $C$ satisfies (A1)-(A4), then $C$ satisfies \textit{Exclusivity}
	\end{lemma}
As we show in Appendix, it is an implication of lemma \ref{exclusivity}.
	\section{Identification} \label{sec6}
	\label{Iden}
There can be multiple representations $(R,P)$ which rationalize a choice function $C$. In this section, we present two results related to identification in the CBR model. Firstly, we define revealed rationales $R^c$ and $P^c$ using the reversals in the choice data. According to our definition, the revealed rationales reflect only those features which are common to every CBR-representation. We then use these rationales to give bounds on both the rationales in the representation. We identify the minimal representation for which the first rationale $R$ is the intersection of first rationales of all the possible CBR representations of $C$. To give the upper bound on the first rationale, we define a revealed rationale which cannot have an intersection with first rationale of any representation. Identification uses a ``small menu property'' of the reversals. All the proofs of this section are relegated to the Appendix.

		\subsection{Small menu property}\label{smp}
It can be shown that any weak reversal in the choice function will be seen in choices from pairs to triples. There will be no binary cycles in the alternatives \textit{involved}\footnote{$x$ is involved in a reversal if either there is a $(xy)$ or $(yx)$ reversal for some $y$ or there is a $(yz)$ reversal due to $x$} in the reversal. Any strong reversal can be seen in either a pair to a triple with a cycle in the pairwise relation, or in a triple to a quadruple with no cycle. 
We define this property as follows:  

\begin{defn}
	A choice function $C$ satisfies \textit{Small Menu Property} (SMP), then the following holds: 
	\begin{itemize}
		\item If there is a \textbf{weak} $(xy)$  reversal due to $z$, then $x \succ_c y \succ_c z$ and $C(xyz)=y$
		\item If there is a \textbf{strong} $(xy)$  reversal due to $z$, then either $x \succ_c y \succ_c z \succ_c x$ cycle exists and $C(xyz)=y$ or $z \succ_c x \succ_c y$ and $C(xyz)=z$ and for some $w$, $C(xyw)=x$ and $C(xyzw)=y$
		
	\end{itemize}
\end{defn}

\begin{lemma} \label{SMP}
	If $C$ is CBR representable, then $C$ satisfies SMP
\end{lemma}

This property enables us to provide an alternative formulation of $\succ_R$ relation in terms of choices from pairs and triples. 

\begin{defn}
	For any $x,y \in X$, $x\succ_R y$ if and only if:
	\vspace{-0.2cm}
	\begin{itemize}
		\setlength\itemsep{0.001cm}
		\item [(i)] $x \succ_c y \succ_c z$ and $C(xyz)=y$ for some $z \in X$; or
		\item [(ii)] $z \succ_c x \succ_c y$ and $C(xyz)=x$ for some $z \in X$; or
		\item [(iii)]$y \succ_c z \succ_c x \succ_c y$ and $C(xyz)=z$; for some $z \in X$ or
		\item [(iv)] $x \succ_c y \succ_c z$, $x \succ_c z \succ_c w$, $C(xwz)=z$ and $C(xyz)=x$ for some $z,w \in X$
	\end{itemize}
\end{defn}
\begin{figure}[h]
	\centering
	\includegraphics[scale=0.25]{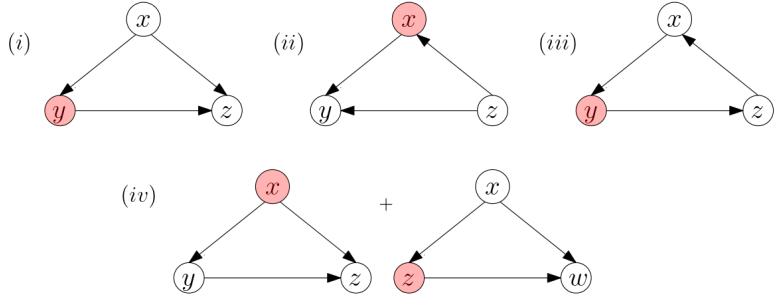}
	\caption{Cases when $x\succ_R y$.  Arrows depict pairwise choices. Colored alternatives are the choices in triples} \label{Reversals}				
\end{figure}
Lemma \ref{SMP} helps us pin down behavior by observing choices over small menus. For any two CBR representable choice functions that agree on \textit{small menus}, i.e. pairs and triples, also agree on larger menus. It is summarized in the result below.
\begin{lemma}\label{Idencor}
	If $C$ and $\bar{C}$ are $CBR$ representable, then $C(\cdot) = \bar{C}(\cdot)$ if and only $C(S) = \bar{C}(S)$ for all $ S \subseteq X$ such that $|S| \leq 3$
\end{lemma} 

\subsection{Class of representations}
We now give a minimal representation of a choice function $C$. A minimal representation has the minimal number of pairwise relations required in the first rationale to rationalize $C$. Formally, it is the intersection of the first rationales of all possible representations. 

We begin by defining a ``revealed" rationale $\tilde{R^c}$. It captures all the information regarding identification that choice reveals about the first rationale. Define
\begin{center}
	$x{R^c}y \iff x\tc y$
\end{center}
As the first rationale is transitive $R^c$ captures the smallest relation that is required for the representation.
Given the first rationale, the second rationale captures those relations which are needed to make the choice from the shortlisted set. An alternative $y$ is shortlisted in a set $S$ if $y \notin \min(S,R^c) $. If $C(S)=x$, then we need $xPy$ for $x$ to be chosen. Hence, we define 
\begin{center}
	$P^c \equiv \hat{P}_{R_c}$ where $x \hat{P}_{R_c}y$ holds if for some $S \subseteq X$, $C(S)=x$ and $y \notin \min(S,R^c) $
\end{center}
Our next result characterizes the entire class of minimal representations in terms of the revealed rationales 
	
	\begin{theorem} \label{Iden1}
		If $C$ is CBR representable and $(R^*,P^*)$ is a minimal representation of $C$, then:
		\begin{itemize}
			\item[(i)] $R^*=R^c$ 
			\item[(ii)] $P^c \subseteq P^* $ where $P^*$ is a complete rationale 
		\end{itemize}
		
	\end{theorem}	
	
Now, we discuss the upper bound on the first rationale recovered from any representation of the model. For this we find out the pairs that cannot be related in any representation.
	
	\begin{defn}
		Given a choice function $C$ define $\hat{Q}$ as  $x\hat{Q}y$ if and only if :
		\begin{itemize}
			\setlength\itemsep{0.001cm}
			\item \textbf{strong} $(xw)$  reversal on set $S$ and $y \in S$ for some $w \in S \setminus \{x\}$
			\item \textbf{strong} $(yw)$  reversal on set $S$ and $x \in S$ for some $w \in S \setminus \{x\}$
			\item \textbf{weak} $(wx)$  reversal and $x,y,w \in S$, $C(S)=w$ for some $w \in S \setminus \{x\}$
		\end{itemize}
	\end{defn}
	
	The largest possible $R$ that can be a part of the representation $(R,P)$ will be the largest transitive relation that is a subset of $\succ_c \setminus \hat{Q} \supset R^c$. We define such largest transitive relation as $\bar{R}$. Note that this largest relation need not be unique. For the second rationale, analogous to $\hat{P}_{R_c}$, we define $x \hat{P}_{R}y$ if for some $S \subseteq X$, $C(S)=x$ and $y \notin \min(S,R) $. Those pairwise relations which are not covered in the first relation are added in the second rationale. The following result provides the class of identified rationales.
	
	\begin{theorem} \label{Iden2}
		If $C$ is CBR representable, then $(R,P)$ represents $C$ if and only if:
		\begin{itemize}
			\item[(i)] $R$ is a transitive rationale such that $R^c \subseteq R \subseteq \bar{R}$
			\item[(ii)] $P$ is a complete rationale such that $P \supseteq \hat{P}_R \cup (\succ_c \setminus R)$
		\end{itemize}
		
	\end{theorem}
	
	\section{Comparison with Related Models} \label{sec7}

	 The violation of rationality (WARP) is attributed to violation of either of the following two consistency conditions: \textit{Always chosen}\footnote{If $x$ is chosen in pairs, then it must been chosen union of those pairs} or \textit{No Binary Cycles}\footnote{Relation derived from pairwise choices cannot have a cycle}. Various boundedly rational models explain violation of rationality using violation of either of these conditions. \cite{manzini2007sequentially} show that RSM is able to accomodate the violation of \textit{No Binary Cycles}. However, a violation of \textit{Always Chosen} cannot be explained by RSM. The ego preserving heuristic (EPH) choice function of \cite{RePEc:ash:wpaper:29} on the other hand is able to accomodate the violations of \textit{Always Chosen} but unable to explain the violation of \textit{No Binary Cycles}. CBR however, is able to explain both the violations.
	 
	 We now compare some related models with CBR and show that if choice function is CBR representable it is equivalent to the related model if strong/weak reversal does not exist.

	\paragraph{(I) Rational Shortlist method:}
	RSM is not a special case of our model. \cite{manzini2007sequentially} characterize it by two axioms: Expansion (EXP) \footnote{For all $S,S' \supset \{x,y\}$, $C(S)= C(S') = x $ implies $C(S \cup S') = x$} and WWARP. Our model may violate WWARP. However, as shown earlier, it satisfies a weaker version of this axiom (R-WARP*) which allows for at most two reversals. Also, CBR may violate EXP as a weak $(xy)$ reversal due to $z$ implies $C(xyz)=y, \ C(xy)=x=C(xz), \ C(yz)=y$ which violates \textit{always chosen}. 
	
	The reversals discussed in this paper establish a relation between our model and RSM. 
	\begin{proposition}\label{RSM}
		If Choice function $C$ is CBR representable, then $C$ is RSM if and only if $C$ has displays no weak reversals
	\end{proposition}
	Proof of this result can be found in the Appendix \ref{proofsec7}.

	\paragraph{(II) Transitive Shortlist method:}
	The transitive shortlist method (TSM) is a variant of the RSM where both the rationales are transitive (possibly incomplete). \cite{horan2016simple} analyzes this choice procedure in terms of two choice reversals:  direct and weak$\st$ \footnote{Weak reversal of TSM. $\star$ added to avoid confusion with weak reversal of this paper} reversal. 
	 A direct $\langle x , y \rangle$ reversal on $B \subset X \setminus \{x\}$ is defined as 
	 \vspace{-0.3cm}
	$$C(B) = y \ \ \text{and} \ \ C(B \cup \{x\}) =z \notin \{x,y\}$$
	A weak $\langle x , y \rangle$ reversal on $B \supset \{x,y\}$ is defined as
	 \vspace{-0.3cm}
	$$C(xy) = x \ \ \text{and} \ \ C(B \setminus \{y\}) \neq C(B)$$
	TSM satisfies Exclusivity condition which says that for a pair $x,y$, either there is no direct $\langle x , y \rangle$ reversal on $B \subset X \setminus \{x\}$ or, there is no weak$\st$ $\langle x , y \rangle$ reversal. CBR violates this axiom when there is a double reversal. It can be seen in Example \ref{psi}. There is a direct $\langle z , x \rangle$ reversal on $\{x,z\}$ and a weak$\st$ $\langle z , x \rangle$ reversal on $\{y,z,w\}$. Another property satisfied by TSM is EXP(hence \textit{always chosen}), which CBR need not satisfy. Thus, TSM is also not a special case of CBR.

	Note that since TSM satisfies WWARP, in the case of a direct $\langle x,y \rangle $ reversal it must be be that $C(yz)=y$. Hence, whenever there is a strong or a weak reversal, we have a direct reversal. 
	Conversely, as TSM also satisfies \textit{always chosen}, a choice function cannot display a weak reversal. Therefore this would be a strong reversal. 

As in the case of RSM, our model relates to TSM in the following way	

	\begin{proposition}\label{TSM}
	If Choice function $C$ is T-CBR representable, then $C$ is TSM if and only if $C$ displays no weak reversals
\end{proposition}

\paragraph{(III) Ego-Preserving Heuristic :} 
\cite{RePEc:ash:wpaper:29} propose a two-stage choice model wherein both the rationales are linear orders. DM first eliminates the worst alternative with respect to the first order and then choose the maximal alternative with respect to the second order. In terms of reversals, our model is related to Ego-Preserving Heuristic (EPH). EPH is characterized by NC, NBC and a weaker form of WARP (WARP-EP). Their model cannot accommodate violation of \textit{no binary cycle}, but allows for violation of always chosen and hence can permit weak reversals. It turns out that their models does not allow for strong reversals as it leads to violation of WARP-EP.

\section{Final remarks} \label{sec8}
In this paper we introduced a new two-stage choice procedure that departs from the idea of shortlisting by maximization. We axiomatically characterized this procedure using intuitive behavioral properties. Bounds on the first and the second rationales were provided to identify the representations for a given choice function. We also compared this procedure with the Rational Shortlist Method (RSM) of \cite{manzini2007sequentially}. The main contribution of our model is its ability to explain double reversals observed experimentally that the existing models are unable to do. In addition to that, our choice procedure also provides an alternative explanation for single reversals discussed in the literature.

The first rationale in our model can be interpreted in several ways. One such interpretation is when alternatives have multiple attributes. DM shortlist those alternatives which are either non-comparable with respect to any attribute, or dominate some alternative with respect to at least one attribute. To illustrate, consider $X= \{x,y,z\}$ and two attributes $R_1 = \{(x,y)\}$ and $R_2 = \{(y,z)\}$. Shortlisting using these two attributes is equivalent to shortlisting by a single transitive rationale $R= \{(x,y), (y,z), (x,z)\}$. Another interpretation is related to ``social influence''. DM is socially influenced by certain reference groups that she relates to: people that she finds similar to herself in a given situation. Pairwise choices of this group are observed, which in aggregate are transitive. Over these choices, she uses her preference to make the final choice.
Given the interpretations above, it is natural to assume that the first rationale need not be complete. 

Our model is also a natural way of choosing in several contexts. One such setting is online dating. Users on a popular dating app, Tinder, are on average presented with 140 partner options a day (\cite{doi:10.1177/1948550619866189}). Large number of partner options sets off a rejection mindset: people become
increasingly likely to reject potential partners before choosing. It may be an interesting future topic to study a possible extension of this model in a stochastic setup. One can think of a collection of rationales and a probability distribution over them that one uses to \textit{reject} ``worse" alternatives before making a final choice.

\pagebreak	
	\appendix
	\section{Appendix}
\subsection{Proof of Theorem 1}

First we prove the necessity of the axioms
	
\begin{proposition}\label{necessity} If $C$ is CBR representable, then it satisfies (A1)-(A4)
		
\end{proposition}

\begin{proof}
Let $(R,P)$ be a representation of the choice function $C$ where $R$ is a partial order and $P$ is a complete rationale. We will use the following three observations and a lemma to prove necessity:
		
\begin{enumerate}
\item  $x R y$ implies $x \succ_c y$. Also, $\succ_c \ \subseteq R\cup P$
			
\item  If $x \notin \min(S,R)$, then either $x$ is \textit{isolated} in $S$ with respect of $R$ ($(x,a) , (a,x) \notin R$  for all $a \in S$) or there exists a $b \in S$ such that $xRb$ holds
			
\item If $x \in \min(S_{i},R)$ for all $i \in [n]$, then $x \in \min(\bigcup_{i} S_{i},R)$ 
			
\end{enumerate}
Now we use the above observations to prove an intermediate result.
		
\begin{lemma} \label{Rev} If $C$ is CBR representable, then the following is true:
\begin{itemize} 
	\item If there is a weak $(xy)$ reversal due to $z$, then $xRy$, $yPx$ and $yRz$
	\item If there is a strong $(xy)$ reversal due to $z$ , then $\neg xRy$, $xPy$, $zRx$ and $yPz$ 
\end{itemize}
\end{lemma}	
\begin{proof}
Let there be a weak $(xy)$ reversal due to some $z$. By observation (1), $ \neg yRx$ and $ x \succ_c y$ implies that $xRy$ or $xPy$.  Suppose $xPy$ holds. Since $C(S)=x$ and $C(S \cup \{z\})=y$, it must be that $x \in \min(S \cup \{z\},R)$ and $x \notin \min (S,R)$. Therefore, we must have $zRx$, contradicting $x \succ_c z$. Thus $xRy$ holds and $x \notin \min(S \cup \{z\},R)$ implying $yPx$. For $C(S)= x$, it must be that $y \in \min(S,R)$ and for $C(S \cup \{z\}) =y$,  $\ yRz$  must be true.  \\
Now, let us consider the case of a strong $(xy)$ reversal due to $z$. If $xRy$ holds, then by the argument above, $yPx$ and $yRz$ holds. By transitivity of $R$, $xRz$ holds which contradicts $z \succ_c x$. Therefore $xPy$ holds and $x$ and $y$ are not related with respect to $R$. For $C(S \cup \{z\})=y$, it must be that $x \in \min(S\cup \{z\},R) $ and therefore for $C(S)=x$, it must be isolated in $S$ with respect to $R$. By an analogous argument in the case above, $zRx$ and $yPz$ hold. 	
\end{proof}

We can see that the following result immediately follows from the lemma above.
\begin{corollary} \label{Cor}
	If $C$ is CBR representable, then $x \ \succ_R \ y \implies xRy$
\end{corollary}	

Now we establish necessity of the axioms
\begin{itemize}

\item[(i)]\textit{NC}: 

For any $S$ with $C(S)=x$, it must be that $x \notin \min(S,R)$. Therefore, either $xRz$ holds for some $z \in S$ or $x$ is isolated in $S$ with respect to $R$. If $xRz$ holds, then we know that $C(xz)=x$. If $x$ is isolated in $S$ with respect to $R$, then we must have at least one $z \in S$ such that $z \notin \min(S,R)$. Therefore we must have $xPz$. Since $x$ and $z$ are unrelated in $S$, we get $C(xz)= x$. 
			
\item[(ii)] \textit{WCC$^*$} : \\
Let $S=\{x,y,x_{1},x_{2},....,x_{n}\}$ and $C(xy)=x$. Define a general set $S_{i}$ which has alternative $x_{i}$ missing from set $S$ i.e. $$S_{i}=S\setminus\{x_{i}\}$$
Assume for contradiction that $C(S_{i}) \notin \{x,y\}$ for all $i \in \{1,2,...,n\}$. Hence, $C(S_{i})$ is one of the $x_{j}$ where $i \neq j$. We denote by $c_{i}$ as the choice in set $S_{i}$.
			
Consider the first case where $C(S)=x$.\\
If $xRy$ then $x \notin$ $\min(S_{i},R)$ for all $i$. For $c_{i}$ to be chosen in $S_{i}$, $c_{i}Px$ must hold for all $i$. Note that for $C(S)=x$, it must be that $c_{i} \in \min(S,R)$ for all $i$, which is possible when $c_{i}$ is isolated in $S_i$ with respect to $R$ and $x_{i}Rc_{i}$ for all $i$. But, for every $i$, there exists a $j \neq i$ such that $c_{i}=x_{j}$, implying that there exists at least one $c_{i} \notin$ $\min(S,R)$ which is a contradiction.\\
Now, let $\neg xRy$ and thus $xPy$ hold. As $x \notin$ $\min(S,R)$, by observation (3), $x \in$ $\min(S_{i},R)$ in at most one $S_{i}$. If $ x \notin \min(S_{i},R)$ for all $i$, then argument becomes similar to the case above where $xRy$ holds. Assume $x \in$ $\min(S_{n},R)$ ($i=n$ W.L.O.G). For $c_{i}$ to be chosen in $S_{i}$ ($i \neq n$), $c_{i}Px$ holds and for $x$ to be chosen in $S$, $c_{i} \in$ $\min(S,R)$ for all $i \neq n$, for which $c_{i}$ is isolated in $S_i$ and $x_{i}Rc_{i}$ for all $i$. This restricts $c_{i}= x_{n}$ for all $i \neq n$. Also, given $x\in \min(S_{n},R)$, for $x\notin \min(S,R)$, we need $xRx_{n}$. As $c_{i}= x_{n}$, there exists a $z \in S \setminus \{x\}$ such that $x_{n}Rz$ holds which implies $x_{n} \notin \min(S,R)$, again a contradiction.
			
Let us now consider the case $C(S)=y$. Now we have $(xy)$ reversal.
Suppose $xRy$ is true. $yPx$ holds as $x \notin \min(S,R)$ and there exists a $x_{k} \in S$ such that $yRx_{k}$ holds. As $y \notin \min(S_k,R) \text{ for all } k \neq i$, choice of $c_{i}$ in $S_{i}$ requires $c_{i}Py$. Further, as $y$ is chosen in $S$, $c_{i} \in \min(S,R)$  for all $i \neq k$. By arguments above, this requires $x_{i}Rc_{i}$ and $c_{i}= x_{k}$ for all $i \neq k$. For $x_{k}$ to be chosen in $S_{i}$, there needs to be an alternative that is dominated by $x_k$ with respect to $R$, which implies $x_{k} \notin \min(S,R)$, a contradiction.\\
Assuming $\neg xRy$, $xPy$ must hold by observation (1). Choice of $y$ in $S$ requires $x \in \min(S,R)$ i.e. there exists a $x_{k}$ (say $x_{n}$) such that $x_{n}Rx$ holds and for no alternative $z$, $xRz$ is true. By observation (3), it must be that $y \in \min(S_{i},R)$ for atmost one $S_{i}$ (say $S_{k}$). Using arguments mentioned above, $c_{i}Py$ and $x_{i}Rc_{i}$ holds for all $i \neq k$ which restricts $c_{i}= x_{k}$ for all $i \neq k$. For $x_{k}$ to be chosen, there exists an alternative below it in $R$, a contradiction.

\item[(iii)] \textit{NBC$^*$}:\\
This follows from corollary \ref{Cor}
			
\item[(iv)]\textit{R-WARP}: \\
Consider $\{x,y\} \subseteq S,S' \in \mathcal{P}(X)$ and $y \in X$ such that the following is true:	
$$y \notin \min(S, \tc),\ C(S)=x ,\ \text{and} \ x \notin \min(S', \tc)$$	
Consider the case when $y \ \tc \ z$ for some $z \in S$. Then by corollary \ref{Cor}, $y \notin \min(S,R)$. As $C(S)=x$, $xPy$ holds. Now, if $x \ \tc \ w$ for some $w \in S'$, then $x \notin \min(S',R)$. This implies $C(S') \neq y$. Now suppose, $\neg x \ \tc \  w$ for any $w \in S'$. For $C(S')=y$, we need $x \in \min(S',R)$. Suppose that $C(xy)=x$. Using the argument in Proposition \ref{necessity} part $(iii)$, there exists a $z \in S'$, such that there is a strong $(xy)$ reversal due to $z$ as weak reversal implies $xRy$ (lemma \ref{Rev}). By definition $z \succ_R x$ holds, which is a contradiction as for $x \notin \min(S',\tc)$, we need $x \tc w$ for some $w \in S'$. If $C(xy)=y$, then by similar argument, there exists a $w \in S$ such that there is a strong/weak $(yx)$ reversal due to some $w \in S$. If the reversal is weak, then $y \succ_R x$ holds. For $x \notin \min(S',\tc)$, there is a $w'\in S'$ such that $x \ \tc \ w'$ holds, a contradiction. If the reversal is strong, then $yPx$ holds, again a contradiction.

Now consider the case when $y$ is not related to any alternative in $S$ with respect to $\succ_R$. If $x \ \tc \ z$ holds for some $z \in S'$, then the case is similar to the case above when $y \ \tc \ z$ and $\neg x \ \tc \ w$ for any $w \in S'$. Hence, consider the case when $y$ is not related to any alternative in $S'$ with respect to $\succ_R$. W.L.O.G, $C(xy)=x$. We then have a $(xy)$ reversal from $\{x,y\}$ to $S'$. As argues above, there exits a $z\in S$ such that reversal is due to $z$. If the reversal is a weak reversal, then $x \ \tc \ y$ holds which contradicts that $x$ is isolated in $S'$. If the reversal is a strong reversal, then $z \ \tc \ x$ holds, which again is a contradiction to $x$ is islolated in $S'$ with respect to $\tc$. \end{itemize}
	\end{proof}
	Next, we prove the sufficiency part of the proof. Before that we prove some lemmas

		\begin{lemma}
		\label{SR1}
		If $C$ satisfies (A1)-(A4), then a strong  $(xy)$ reversal implies $\neg x \ \tc \ y$
	\end{lemma}
	\begin{proof}
		Given there is a strong  $(xy)$ reversal on $S$ due to some $z$ , by the definition of $\succ_R$ , $z \succ_R x$ holds. If possible, $ x \ \tc \ y$. This implies $z \ \tc \ y$. Note that by R-WARP, $y \in \min(S, \tc)$. We now have a $(zy)$ reversal from $\{y,z\}$ to $S \cup \{z\}$. By WCC$^*$, there exists a $x_1 \in S$ such that the reversal is due to $x_1$. If it is a weak $(zy)$ reversal due to $x_1$, then $y \succ_R x_1$ holds, which is a contradiction. Therefore, it must be that we have strong $(zy)$ reversal due to $x_1$. By definition, $x_1 \succ_R z$ holds. This further implies $x_1 \ \tc \ y$. By NBC$^*$, $C(x_1y)=y$. Now, we have a $(x_1y)$ reversal due to some $x_2 \in S$. By similar argument as above, this must be a strong reversal, implying $x_2 \ \tc \ y$. This leads to $(x_2y)$ reversal due to some $x_3 \in S$. Proceeding inductively, this leads to $x_i \ \tc \ y$ for all $x_i \in (S \cup \{z\}) \setminus \{x\}$ as in each step, $x_{i+1} \neq x_k$, $k \le i$ by NBC$^*$. This violates NC as $x_i \succ_c y$ for all $x_i \in S \cup \{z\}$, a contradiction. 
	\end{proof}
	
	\begin{lemma}
		\label{SR}
		If $C$ satisfies (A1)-(A4), then a strong $(xy)$ reversal on set $S$ implies $ \neg x \ \tc \ w$ and $ \neg  w \ \tc \ x $ for all $w \in S$
	\end{lemma}
	
	\begin{proof}
		Suppose a strong $(xy)$ reversal is observed on set $S$. By lemma \ref{SR1},  $\neg x \ \tc\ \ y$ and by NBC$^*$,  $\neg y \ \tc\ \ x$. Hence, $y \notin \min(\{x,y\}, \tc)$. By R-WARP, $x \in \min(S \cup \{z\}, \tc)$ which implies $\neg x \ \tc \ w$ for all $w \in S$. Therefore $ x \in \min(S, \tc)$. If possible for some $w \in S$, $w \ \tc \ x$ holds. Then by WCC$^*$ and lemma \ref{SR1} there is a weak $(wx)$ reversal due to some $w' \in S$. By definition, $x \succ_R w'$, contradicting $x \in \min(S, \tc)$.  
	\end{proof}

			\begin{lemma}
	\label{exclusivity}
	If $C$ satisfies (A1)-(A4) and there is a \textbf{weak} $(xy)$  reversal, then there does not exist $ y' \in X$ and a menu $S$ such that there is a:
	\begin{itemize}
		\item \textbf{strong} $(xy')$ reversal on $S \ni y$ ; or
		\item \textbf{strong} $(yy')$ reversal on $S \ni x$
	\end{itemize} 
	
\end{lemma}
\begin{proof}
	Let there be a weak $(xy)$ reversal due to some $z$. By definition, $x \succ_R y$ holds. In both the cases, this is a contradicton as this violates lemma \ref{SR}.
\end{proof}
Exclusivity mentioned in section \ref{axiomsdiscussion} is a direct implication of the lemma above.
	 
	We can see that if a choice function that is CBR representable has no strong reversal, then $C$ has no binary cycle in pairwise relation (satisfies \emph{No Binary cycle} condition). Next we show that our axioms imply \textit{Small Menu Property (SMP)}
	
	\begin{lemma}
		\label{Red}
		If $C$ satisfies (A1)-(A4), then it satisfies ``small menu property''
	\end{lemma}
	\begin{proof}
		Suppose for some $x,y \in X$ we have a weak $(xy)$ reversal due to some $z$. By definition $x \ \succ_R \ y, y \  \succ_R \ z $ and hence $x \ \tc \ z$. By NBC$^*$, $x \succ_c y \succ_c z $ and , $C(xyz) \neq z$ due to NC. Now, by R-WARP, $C(xyz) \neq x$. Therefore we have a weak $(xy)$ reversal from pair to triple. \\		
		Suppose for some $x,y \in X$ we have a strong $(xy)$ reversal due to some $z$ on some set $S$. By definition $z \succ_R x$. We know that $\neg y \ \tc \ z $ (as $y \ \tc \ z$  would imply $ y \  \tc \ x$, violating NBC$^*$). If $C(xyz)= x$ , then we have a $(zx)$ reversal due to $y$. It cannot be a strong reversal as it would imply $y \ \tc \ z$. Therefore, it must be a weak $(zx)$ reversal due to $y$ which implies $x \ \tc \ y$, a contradiction to lemma \ref{SR1}. If $C(xyz) = y$, then by NC, $C(yz)=y$ and we are done. 
		
		Now, let us consider $C(xyz) = z$. Since $C(S \cup \{z\})=y$ and $z \notin \min(S \cup \{z\}, \tc)$, thus by R-WARP, $ y \in \min(\{x,y,z\}, \tc ) $. This implies $z \ \tc \ y$ as $\neg x \ \tc \ y$ by lemma  \ref{SR1}. By NBC$^*$, $C(yz)=z$. Since $C(S \cup \{z\})=y$, by WCC$^*$ there exists a $w \in S$ such that the $(zy)$ reversal is due to $w$. This reversal is a weak reversal by lemma \ref{corr} and \ref{SR} and hence $y \ \tc \ w$. This implies $z \ \tc \ w$ (and $C(zw)=z$ by NBC$^*$). Now, let us show that $C(xyw)=x$. By R-WARP, $C(xyw) \neq y$ since $x \notin \min(\{x,y,w\}, \tc)$. If $C(xyw)=w$ then by lemma \ref{corr} and \ref{SR1}, we have a weak $(yw)$ reversal due to $x$ implying $y \ \tc \ x$ , a contradiction. Now, we show $C(xyzw)= y$. By R-WARP, $C(xyzw) \neq z$ as $y \notin \min(xyzw, \tc)$. Also, if $C(xyzw)=x$, then we have a weak $(zx)$ reversal due to $w$. By definition, $x \succ_R w$, violating lemma \ref{SR}. If $C(xyzw)= w$, then by WCC$^*$, the $(yw)$ reversal is either due to $x$ or $z$. It is not due to $x$ by lemma \ref{SR}. If the reversal is due to $z$ then this is a weak reversal, contradicting $C(zy) = z$.                   
	\end{proof}
	\bigbreak
	We are now equipped to prove the sufficiency of the axioms.
	
	\begin{proposition}\label{prop3}
		If $C$ satisfies A(1)-A(4), then it is CBR representable
	\end{proposition} 
	\begin{proof}  
		First, we define a partial order $R^c$ and a complete rationale $P^c$ on the choice function such that $$ C(S)= \max(S \setminus \min(S,R^c), P^c) $$ 
		
	Define $R^c \equiv \tc $. To define $P^c$, we first define $P_1$ as 
		\begin{center}
			$xP_1 y$ if and only if there exists a $S$ such that $C(S)=x$ and $y \notin \min(S, \tc)$
		\end{center}

Note that $ \succ_R$ is acyclic by NBC$^*$. Hence, $R^c$ is asymmetric and transitive. $P_1$ is asymmetric by R-WARP. Let $P^c$ be an asymmetric and complete rationale such that	$$P^c \equiv P_1 \cup P_2$$ where $P_2=R^c \setminus P_1 \cup P_1^{-1}$
	
	Let us now prove that $P^c$ is complete and asymmetric.\\
	
	\textbf{Completeness:} Let us assume that $C(xy)=x$. If $x$ and $y$ are not related with respect to $R^c$, this implies $xP_1y$ since $y \notin \min(xy, R^c)$. Now, if  $(xy) \notin P_1 \cup P_1^{-1}$, then it is related with respect to $R^c$. Therefore, $x$ and $y$ are related with respect to $P_2$

	\textbf{Asymmetry:} If possible, say for some $x,y$, both $xP^c y$ and $yP^c x$ is true. Either $x{P_1} y$ or $y{P_1} x$ is true otherwise both will be derived through $P_2$ which contradicts the asymmetry of $R^c$. W.L.O.G. suppose $x{P_1} y$ holds. Then $y{P_1}x$ cannot hold by R-WARP. It is easy to see that $y P_2 x$ also does not hold.
		
		Now, we show that the above defined $(R^c, P^c)$ rationalize the choices. 
		Consider a set $S$ and $C(S)=x$. Suppose that $x \in \min(S,R^c)$. Then there exists a $y \in S \setminus \{x\}$ such that $yR^cx$ holds and $ \neg xR^cy'$ for all $ y' \in S \setminus \{x\}$.  Note that by WCC$^*$, there exists a sequence of sets ordered by set inclusion from $\{x,y\}$ to $S$ with choices belonging to $\{x,y\}$. R-WARP ensures that there exists a $z \in S$ such that there is a $(yx)$ reversal due to $z$. Lemma \ref{corr} and \ref{SR} implies that it is a weak reversal. By Lemma \ref{Red}, $y \succ_c x \succ_c z$ and $C(xyz)=x$. Thus, we must have $x R^c z$, leading to a contradiction.  \\
		Now we show that $x = \max(S \setminus \min(S,R^c),P^c)$. Consider any $y$ such that $y \notin \min(S, R^c)$ and $ y P^c x$ holds. We know that by construction of $P^c$, we have $x {P_1}y \  (\implies xP^cy)$ which contradicts the asymmetry of $P^c$. Therefore $xP^cy$ for all $y \notin \min(S,R^c)$. 	
	\end{proof}
	
	Proposition \ref{necessity} and Proposition \ref{prop3} complete the proof of Theorem \ref{thm1} 	

	\subsection{Proof of Theorem 2}
	\begin{proposition} \label{prop4}
		If $C$ is Transitive CBR-representable, then $C$ satisfies A(1), A(2),(A3),(A4') 
	\end{proposition}
	\begin{proof}
		The necessity of \textit{WCC$^*$}, NC and \textit{NBC$^*$} is same as shown in Appendix A.1. Let us now prove the necessity of R-SARP. Suppose for some  $S_{1}, \ldots, S_{n} \in \mathcal{P}(X)$ and distinct $x_{1}, \ldots, x_{n} \in X$, we have:
$$	x_{i+1} \notin \min(S_i, \tc), \ C(S_{i})=x_{i} \text{ \ for \ }  i=1, \ldots, n-1, \text { and } x_{1} \notin \min(S_n, \tc\textit{}) 
		$$
		Using the argument in proving the necessity of R-WARP, we must have $x_iPx_{i+1}$ for all $i$. Since $P$ is transitive, we must have $x_1Px_n$. If $C(S_n) =x_n$, by a similar argument, it would imply $x_nPx_1$ a contradiction (since $P$ is asymmetric)
	\end{proof}
	Note that all the lemmas in the section above (Lemmas \ref{SR1} - \ref{Red}) hold true even when R-WARP is replaced by R-SARP. With this we prove the following result
	\begin{proposition} \label{prop5}
		If $C$ satisfies A(1), A(2),(A3),(A4'), then it is Transitive-CBR representable
	\end{proposition}
	
	\begin{proof}	Define $R^c \equiv \tc$ and $P^c$ as
		\begin{equation}
		\label{PcTCBR}
		P^c \equiv \bar{P_1} \cup \hat{P_2}
		\end{equation}
	
		 where $xP_1 y$ if and only if there exists a $S$ such that $C(S)=x$ and $y \notin \min(S, \tc)$
 and $\bar{P_1} = tc(P_1)$. Also, $\hat{P_2} \equiv R^c \setminus(\bar{P_1} \cup \bar{P_1}^{-1})$
		 
		 $R^c$ is asymmetric and transitive as discussed above. Next we show that $P^c$ is a linear order.
		
		\textbf{Completeness:} Let us assume that $C(xy)=x$. If $x$ and $y$ are not related with respect to $R^c$, this implies $xP_1y$. Now, if  $(xy) \notin \bar{P_1} \cup \bar{P_1}^{-1}$, then it is related with respect to $R^c$. Therefore, $(xy)$ is related with respect to $\hat{P_2}$
				
		\textbf{Asymmetry:} If possible, say for some $x,y$, both $xP^c y$ and $yP^c x$ is true. Either $x\bar{P_1} y$ or $y\bar{P_1} x$ is true (else both will be derived through $R^c$ which is a contradiction). W.L.O.G suppose $x\bar{P_1} y$ holds. Then $y\bar{P_1}x$ cannot hold by R-SARP. Thus, it must be that $\neg y P^c x$ as it would mean $yR^cx$ and $x,y$ are not related with respect to $\bar{P_1}$.
		
		\textbf{Transitivity:} Assume for contradiction that $P^c$ is cyclic. Since $P^c$ is complete, we only need to consider a 3-cycle i.e. for some $x,y $ and $z$, $xP^cyP^czP^cx$. It is easy to see that at least one of the pair must be related in $\bar{P_1}$, thus the following cases are possible:
		\begin{itemize}
			\item $ x \bar{P_1}y$ and $y \bar{P_1}z$: This would imply $x \bar{P_1}z$ and  $\neg z\bar{P_1}x$ therefore $\neg zP^c x$.
			\item $ x \bar{P_1}y$   and $y \hat{P_2} z$: Since $zP^cx$ is true, following two cases are true: 
			\begin{itemize}
				\item $z \bar{P_1} x$: This would imply $z \bar{P_1}y$ and  $\neg y\bar{P_1}z$ therefore $zP^c y$, a contradiction to the asymmetry of $P_c$.
				\item $z \hat{P_2}x$ : By \textit{NBC$^*$}, we know that $y R^c x$. By NC and NBC$^*$, $C(xyz) \in \{y,z\}$. As $ y, z \notin \min(xyz, R^c)$,  we get $ (yz) \in \bar{P_1} \cup \bar{P_1}^{-1}$, a contradiction.

			\end{itemize}
		\end{itemize}

		Now, we show that the above defined $(R^c,P^c)$ rationalize the $C$. Consider a set $S$ and $C(S)=x$. Suppose $x \in \min(S,R^c)$. Then there exists $y \in S \setminus \{x\}$ such that $yR^cx$ holds and $ \neg xR^cy'$ for all $ y' \in S \setminus \{x\}$. Note that by WCC$^*$, there exists a sequence of sets ordered in set inclusion from $\{x,y\}$ to $S$ with choices belonging to $\{x,y\}$. R-SARP ensures that there exists a $z \in S$ that causes the $(yx)$ reversal. As argued above, Lemma \ref{corr} and \ref{SR} imply that the reversal is weak. By Lemma \ref{Red}, $y \succ_c x \succ_c z$ and $C(xyz)=x$. Thus, $x \ \tc \ z$, leading to a contradiction. 
	
	Now we show that $x = \max(S \setminus \min(S,R^c),P^c)$. Consider any such $y$ that $y \notin \min(S, \tc)$ and $ y P^c x$ holds. We know that by construction of $P^c$, we have $x P_1 y (\implies x P^c y)$ holds  which contradicts the asymmetry of $P^c$. Therefore $x P^c y$ for all $y \notin \min(S, \tc)$.

\end{proof}	
Proposition \ref{prop4} and \ref{prop5} prove the sufficiency.	
	
\subsection{Proofs of results in Section \ref{axiomsdiscussion}} \label{proofsec4.2}
	\subsubsection{Proof of Lemma \ref{rwwarp}}

		Suppose for some $S,S'$ and $S''$ we have $C(S) = C\{x,y\} = x$ and $C(S') = y$. Note that we have a $(xy)$ reversal. By WCC$^*$ there exists a $z \in S$ that causes this reversal. If $C(xz)=x$ then it is a weak $(xy)$ reversal. By definition, $x \ \tc \ y $ and $ y \ \tc \ z$, implying $C(S'') \neq y$ by R-WARP. If $C(xz)=z$, this implies a strong $(xy)$ reversal. Also since $z \ \tc \ x$ there is a weak $(zx)$ reversal due to some $z' \in S$ (by WCC$^*$). Therefore $x \notin \min(S', \tc)$, implying $C(S') \neq y$ by R-WARP

	\subsection{Proofs of results in Section \ref{sec6}}	
\subsubsection{Proof of Lemma \ref{SMP}}

Follows from proposition \ref{prop3} and lemma \ref{Red}.

\subsubsection{Proof of Lemma \ref{Idencor}}

	Using the lemma \ref{SMP}, we can say that if we have either a weak or a strong reversal, then it will be reflected in the small menu reversals i.e. sets such that $|S| \le 3$. Now consider two choice functions $C$ and $\bar{C}$ with same choices in small menus, but different choice in at least one set $S$ where $|S| > 3$. W.L.O.G., let the choice in that set be $C(S)=x$ and $\bar{C}(S)=y$ where $x \neq y$. Let $C(xy)=x$. Hence, we have a $(xy)$ reversal in the choice function $\bar{C}$. We have argued before that any reversal in a CBR representable choice function can either be a weak or a strong reversal. By lemma \ref{Red}, this reversal will be reflected in the small menus (hence relations required will be common to both the representations). If the reversal is weak, then it will be reflected in small menus giving $yR^*z$ and $yP^*x$. Hence, $x$ cannot be chosen in a set containing $y$ and $z$ which is contradiction as $C(S)=x$. If the reversal is strong, then if we two possible cases: 
	
	(i) $x \succ_c y \succ_c z \succ_c x$ and $C(xyz)=y$. Here the $(xy)$ reversal is seen in a small menu giving $xP^*y$ and $\neg xR^*y$, $yP^*z$, $zR^*x$ . Note that $C(xyz)=y$ and $C(S)=x$ where $S \supset \{xyz\}$. As $C(xz)=z$, we have a $(zx)$ reversal due to some $w \in S$. Knowing that $zR^*x$ is true, this is a weak reversal, which is seen in small menus implying $xR^*w$ holds. This means $x \notin \min(S, \bar{R})$ contradicting $\bar{C}(S)=y$. 
	
	(ii) $z \succ_c x \succ_c y$, $z \succ_c y \succ_c w$, $C(xyz)=z$ and $C(yzw)=y$ for some $w \in X$. Note that again we have a weak $(zx)$ reversal due to some $k \in S$, leading to a contradiction as above.

\subsubsection{Proof of Theorem \ref{Iden1}}
\begin{itemize}
	\item[(i)] Note that lemma \ref{Rev} and proposition \ref{prop3} imply that $R^c \subset R^*$ . Also note that proof of Proposition \ref{prop3} rationalizes choice function using $R^c$ as the first rationale. This proves that $R^* \subset R^c$. Hence, $R^* = R^c$.
	\item[(ii)] As the choice procedure chooses the maximal alternative from the set of alternatives which do not belong to $\min(S,R)$, $\hat{P}_{R^c}$ is a subset of any $P$ such that $(R^c,P)$ is a representation. For all the pairs which are not related with respect to $R^c$, both the alternatives are shortlisted. Since, $x \succ_c y$, $xPy$ is true, $P^c \subseteq P^*$  
\end{itemize}
\subsubsection{Proof of Theorem \ref{Iden2}}
\begin{itemize}
	\item[(i)] We first show that $\hat{Q} \cap R =\phi$ for any $R$ such that $(R,P)$ represents choice function. Consider the first case with a strong $(xw)$ reversal on $S \ni y$ due to say, $z$. By lemma \ref{Rev} and proposition \ref{prop3}, we know that $zR^*x$, $xP^*w$ and $wP^*z$ holds. For $w$ to be chosen in $S \cup \{z\}$, it must be that $x in \min(S\cup \{z\},R)$ therefore we cannot have $xRy$. Similarly, in the second case, we have $(yw)$ strong reversal on $S \ni x$ due to say, $z$. This implies that $zR^*y$, $yP^*w$ and $wP^*z$ must hold (lemma \ref{Rev} and proposition \ref{prop3}). For $w$ to be chosen in $S \cup \{z\}$, it must be that $y \in \min(S \cup \{z\})$. Now, if $xRy$ holds, then there must exist a $a \in S$ such that $ yRa$ (implying $y \notin \min(S \cup \{z\})$) holds, which contradicts the assumption that $C(S \cup \{z\}) =w$ . In the final case, where we have weak $(wx)$ reversal and $w$ is chosen in the presence of $x,y$. We know that $yP^*w$ holds. For $w$ to be chosen in presence in the presence of $y$, $y$ should not be shortlisted, which is not possible with $xRy$ by the same argument as in the previous case. \\
	We have proved above that $R^c$ is the smallest possible $R$. $\bar{R}$ must be a subset of pairwise relation $\succ_c \setminus \hat{Q}$. As $R$ is a transitive relation, $\bar{R}$ must be the largest transitive relation which is a subset of $\succ_c \setminus \hat{Q}$ 
	
	\item[(ii)] The argument is similar to that of $P^*$ in proof of theorem above, replacing $R^c$ with $R$.
\end{itemize}
	\subsection{Proof of results in Section \ref{sec7}} \label{proofsec7}
	Before we begin the proof, we prove some intermediate results
\begin{defn}
	Negative Expansion (NE): For all $S,S' \supset \{x,y\}$,
	\begin{center}
		$C(S)= C(S') = x $ implies $C(S \cup S') \neq y$
	\end{center}
	
\end{defn}
\begin{lemma}
	\label{NE}
	If $C$ is CBR representable, then $C$ satisfies Negative Expansion
\end{lemma}
\begin{proof}
	If possible, suppose choice function $C$ violates NE. Then there exists $S, S' \supset \{x,y\}$ such that $C(S)=C(S')=x$ and $C(S \cup S')=y$. If $C(xy)=x$, then we have an $(xy)$ reversal. By Proposition \ref{necessity}, $C$ satisfies A(1)-A(4) implying each reversal is either weak or strong. Let $(R,P)$ be the representation of the choice function. A weak $(xy)$ reversal due to $z \in S$ and $z' \in S'$ implies $yPx$, $yRz$ and $yRz'$ (by lemma \ref{Rev}), therefore it must be that $y \notin \min(S,R)$. This implies $C(S) \neq x$. Also, if there is a strong $(xy)$ reversal due to some $z \in S \cup S'$ , then lemma \ref{Rev} implies $zRx$. Given $C(S)=x$, we know that $x \notin \min(S, R)$ and there is a $w \in S$ such that $xRw$ holds. As $x \notin \min(S \cup S',R)$ and by lemma \ref{Rev}, we know that  $xPy$ holds. This contradicts $C(S \cup S')= y$. \\
	Now, if $C(xy)=y$, we have an $(yx)$ double reversal. A weak $(yx)$ reversal (due to some $z \in S$ and $z' \in S'$) implies $xRz$ and $xRz'$ (by lemma \ref{Rev}). Given $xPy$ and $x \notin \min(S \cup S', R)$, we cannot have $C(S \cup S')=y$. If the $(yx)$ reversal is strong, then by lemma \ref{Rev}, $yPx$ must hold. Since $C(S)=C(S')=x$, we must have $y \in \min(S,R)$ and $y \in \min(S',R)$ (thus $y \in \min(S \cup S')$) contradicting $C(S \cup S')=y$ 
\end{proof}
\bigbreak
		\begin{lemma}\label{dbl}
	If $C$ is CBR representable, then a $(xy)$ double reversal due to $z_1,z_2$ is equivalent to a strong $(xy)$ reversal due to $z_1$ and a weak $(z_1x)$ reversal due to $z_2$
\end{lemma}
\begin{proof}
	Let $C$ be a CBR representable choice function. A $(xy)$ double reversal due to $z_1$, $z_2$ implies $x \succ_c y$ and $\exists \ S,S'$, $\{x,y\} \subset S' \subset S$ such that $C(S')=x,C(S' \cup z_1)=y \ C(S)=y, C(S \cup z_2)=x$ and for all $T,T', T''$, $\{x,y\} \subset T' \subset T \subset T''$, if $C(T) = C\{x,y\} = x$ and $C(T') = y$, then $C(T'') \neq y$. As $C$ satisfies WCC$^*$ and  Exclusivity, each reversal is either weak or strong. If first $(xy)$ reversal is weak, then lemma \ref{Rev} implies $xRy$, $yRz_1$ and $yPx$. As $x,y \notin \min(T,R)$ for all $T \supset S$, there can be no double reversal. Thus, first reversal is strong, implying $z_1Rx$, $xPy$ and $yPz_1$. For $x$ to be chosen again in $S \cup \{z_2\} $, it must be that $xRz_2$ and $xPz_1$ hold. This implies $z_1 \succ_c x \succ_c z_2$ and $C(xz_1 z_2)=x$ and hence a weak $(z_1 x)$ reversal due to $z_2$ 
\end{proof}
	\subsubsection{Proof of Proposition \ref{RSM}}

		Let us prove the if part. Consider a CBR representable choice function $C$ which is also RSM representable. If possible, for some $x,y$ we have a weak $(xy)$ reversal due to some $z$. By Lemma \ref{Rev}, we have $x \succ_c y \succ_c z$ and $C(xyz)=y$. However, this violates Expansion as $C(xy)=C(xz)=x$, but $C(xyz)=y$. As RSM satisfies Expansion, this is a contradiction.
		
		Now consider $C$ which is CBR representable, with no weak reversals. By proposition \ref{necessity}, if there is any reversal, it has to be strong. If possible, let $C$ violate Expansion, i.e. there exists $S, S'$ such that $C(S)=C(S')=x$, but $C(S \cup S') =y \neq x$. If $\{x,y\} \subset S \cap S'$, this violates NE leading to a contradiction (by lemma \ref{NE}). WLOG, let $y \in S \setminus S'$. If $C(xy)=y$, we have double reversal which is a contradiction by lemma \ref{dbl}. Thus, $C(xy)=x$ implying a $(xy)$ strong reversal due to some $z \in S'$. By lemma \ref{Rev}, $xPy$ and $zRx$ hold and since $C(S')=x$, there exists a $w \in S'$ such that $xRw$ is true. Note that this implies $x \notin \min(S \cup S',R)$ which implies $C(S \cup S') \neq y$. Now, if $C$ violates WWARP, given it satisfies R-WARP, there is a double reversal. But, that is equivalent to a strong and a weak reversal which is a contradiction. Thus, $C$ satisfies Expansion and WWARP, implying RSM representation.

	\subsubsection{Proof of Proposition \ref{TSM}}

		Argument is analogous to that of proposition \ref{RSM}

	\subsection{Independence of axioms}
	By means of simple examples, we demonstrate that the characterization is tight. \\
	
	Example 1. The choice function below satisfies NBC$^*$, WCC* and R-WARP but violates NC: $X = \{x,y,z\}$
{\footnotesize	
	\begin{center}
		\begin{tabular}{|c|c|c|c|} 
			
			\hline 
			
			{\textbf{S}} & {\textbf{C(S)}} & {\textbf{S}} & {\textbf{C(S)}} \\
			
			\hline 
			
			$\{x,y\}$ & $y$ & $\{x,y,z\}$ & $x$ \\
			$\{x,z\}$ & $z$  & &\\
			$\{y,z\}$ & $y$ & &\\
			\hline 
		\end{tabular}
	\end{center} 
}	
	
	\vspace{.5cm}
	
	Example 2. The choice function below satisfies NBC$^*$,  NC and R-WARP but violates WCC*: $X = \{x,y,z, w\}$ {\footnotesize
	\begin{center}
		\begin{tabular}{|c|c|c|c|c|c|} 
			
			\hline 
			
			{\textbf{S}} & {\textbf{C(S)}} & {\textbf{S}} & {\textbf{C(S)}} & {\textbf{S}} & {\textbf{C(S)}} \\
			
			\hline 
			$\{x,y\}$ & $x$ & $\{x,y,z\}$ & $z$ & $\{x,y,z,w\}$ & y \\
			$\{x,z\}$ & $z$ & $\{x,y,w\}$ & $w$ &  & \\
			$\{x,w\}$ & $w$ &  $\{x,z,w\}$ & $w$ &  &\\
			$\{y,z\}$ & $z$ & $\{y,z,w\}$ & $y$ & & \\
			$\{y,w\}$ & $y$ &  &  & & \\
			$\{z,w\}$ & $w$ &  &  & & \\
			\hline 
		\end{tabular}
		
	\end{center} 
}
	\vspace{.5cm}

	Example 3. The choice function below satisfies NC, WCC* and R-WARP but violates NBC$^*$: $X = \{x,y,z, w\}$ {\footnotesize
	\begin{center} 
		\begin{tabular}{|c|c|c|c|c|c|} 
			\hline	
			{\textbf{S}} & {\textbf{C(S)}} & {\textbf{S}} & {\textbf{C(S)}} & {\textbf{S}} & {\textbf{C(S)}} \\
			\hline 
			$\{x,y\}$ & $x$ & $\{x,y,z\}$ & $y$ & $\{x,y,z,w\}$ & $y$ \\
			$\{x,z\}$ & $z$ & $\{x,y,w\}$ & $y$ &  & \\
			$\{x,w\}$ & $x$ &  $\{x,z,w\}$ & $x$ &  &\\
			$\{y,z\}$ & $y$ & $\{y,z,w\}$ & $z$ & & \\
			$\{y,w\}$ & $y$ &  &  & & \\
			$\{z,w\}$ & $z$ &  &  & & \\
			\hline 
		\end{tabular}
		
	\end{center}
}
	\vspace{.5cm}

	Example 4. The choice function below satisfies NC, WCC* and NBC$^*$ but violates R-WARP: $X = \{x,y,z, w\}$ 
	{\footnotesize
	\begin{center} 
	\begin{tabular}{|c|c|c|c|c|c|} 
		\hline 	
		{\textbf{S}} & {\textbf{C(S)}} & {\textbf{S}} & {\textbf{C(S)}} & {\textbf{S}} & {\textbf{C(S)}} \\
		\hline 
		$\{x,y\}$ & $x$ & $\{x,y,z\}$ & $y$ & $\{x,y,z,w\}$ & $x$ \\
		$\{x,z\}$ & $x$ & $\{x,y,w\}$ & $x$ &  & \\
		$\{x,w\}$ & $x$ &  $\{x,z,w\}$ & $x$ &  &\\
		$\{y,z\}$ & $y$ & $\{y,z,w\}$ & $y$ & & \\
		$\{y,w\}$ & $y$ &  &  & & \\
		$\{z,w\}$ & $z$ &  &  & & \\
		\hline 
	\end{tabular}
\end{center}
	}
	
	\noindent

	\newpage

	\setlength{\bibsep}{0.2 cm} 
	\bibliographystyle{ecta} 
	\bibliography{bibliog}

\end{document}